\documentclass[11pt,reqno]{amsart}

\usepackage{amssymb,mathtools,amsmath}
\usepackage{url}
\usepackage{hyperref}
\usepackage[dvips]{color}
\usepackage{tikz-cd}
\usepackage[all]{xy}
\usepackage{cite}
\SelectTips{cm}{} 

\allowdisplaybreaks[4]

\usepackage{mathrsfs}
\let\mathcal\mathscr
\usepackage[mathscr]{eucal}
\usepackage{mathtools}

\newcommand*{\pd}[2]{\mathchoice{\frac{\partial#1}{\partial#2}}
  {\partial#1/\partial#2}{\partial#1/\partial#2}
  {\partial#1/\partial#2}}

\let\phi=\varphi
\let\kappa=\varkappa
\let\epsilon=\varepsilon

\DeclareMathOperator{\sym}{sym}

\newcommand*{\jac}[2]{\left\{ #1,#2  \right\}}

\newcommand\restr[2]{{
  \left.\kern-\nulldelimiterspace 
  #1 
  \littletaller 
  \right|_{#2} 
  }}

\newcommand{\littletaller}{\mathchoice{\vphantom{\big|}}{}{}{}}

\newdir{ >}{{}*!/-5pt/\dir{>}}

\theoremstyle{theorem}
\newtheorem{proposition}{Proposition}

\theoremstyle{definition}

\theoremstyle{remark}
\newtheorem{remark}{Remark}

\usepackage{mathrsfs}
\let\mathcal\mathscr
\usepackage[mathscr]{eucal}
\newcommand{\cprime}{\/{\mathsurround=0pt$'$}}

\author{Petr~Voj{\v{c}}{\'{a}}k}
\address{Mathematical Institute, Silesian University in Opava, Na Rybn\'{\i}\v{c}ku 1, 746 01 Opava, Czech Republic}
\email{Petr.Vojcak@math.slu.cz} 
\title[Nijenhuis torsion and Fr\"olicher--Nijenhuis brackets of recursion operators]{Nijenhuis torsion and Fr\"olicher--Nijenhuis brackets of recursion operators via their full-fledged forms}

\begin{document}
\maketitle

\begin{abstract}
We present a novel approach to computing the Nijenhuis torsion and Fr\"olicher--Nijenhuis brackets of recursion operators for symmetries, based on their full-fledged forms introduced by Jahnov\'a and Voj\v c\'ak (2024). In contrast to the conventional approach, which represents recursion operators as maps between shadows of nonlocal symmetries within a given covering, our method allows the Nijenhuis torsion to be computed directly from its defining formula, without requiring any additional mathematical constructions. The same framework can also be used to compute the Fr\"olicher--Nijenhuis bracket of recursion operators and to verify their compatibility directly in their full-fledged forms. The procedure is illustrated by several examples of full-fledged recursion operators for symmetries of differential equations in four independent variables, including a new recursion operator for the four-dimensional universal hierarchy equation that, to the best of our knowledge, has not previously appeared in the literature. A notable feature of this new full-fledged recursion operator is that its shadow component does not appear to admit a reasonable representation in any of the standard conventional forms. Nevertheless, our approach allows us to prove directly that the Nijenhuis torsion of this operator vanishes. For the examples considered, the corresponding Fr\"olicher--Nijenhuis brackets also vanish, confirming the compatibility of the respective pairs of recursion operators. This demonstrates that full-fledged forms provide an effective framework for studying the hereditary and compatibility properties of recursion operators, including highly nonlocal and/or multidimensional cases that are difficult to handle by existing methods.
\end{abstract}

\section{Introduction}
The construction of recursion operators for symmetries and the subsequent study of their properties have for almost fifty years constituted one of the central problems in the theory of integrable systems of partial differential equations (PDEs). Their knowledge allows one to generate infinite hierarchies of symmetries of a given system in a relatively straightforward manner, which can subsequently be used, for instance, in the search for its explicit solutions. Since their introduction by Olver in 1977 \cite{Olv77}, recursion operators have been extensively studied by many authors, see, e.g., \cite{Blu89, Fok87, Olv93} and references therein.

Within this classical approach, recursion operators are typically constructed as pseudodifferential operators, consisting of linear combinations of total derivatives and their formal inverses, with coefficients given by differential functions. A prototypical example of such an operator is the Lenard recursion operator \cite{Gar74, Olv77}
\begin{equation}
\label{lenard}
\mathcal R = D_x^2+4u+2u_xD_x^{-1},
\end{equation}
which generates symmetries of the well-known Korteweg--de Vries (KdV) equation
$$u_t = 6uu_x+u_{xxx},$$
see also \cite{Kra11}.

Since recursion operators conventionally act on the generating functions of symmetries, we briefly recall this standard setting; for a more detailed theoretical exposition, see, e.g., \cite{Olv93, Boch99}, as well as the overview paper \cite{Kra11}. Let
\begin{equation}
\label{syst_E}
\mathcal E: F_i\left( \mathbf x, \mathbf u, \mathbf p\right)= 0, \quad i=1,\ldots,k,
\end{equation}
be a system of differential equations, where $\mathbf u = (u^1, \ldots, u^m)$ are the unknown functions of the independent variables $\mathbf x=(x_1, \ldots, x_n)$, $\mathbf p$ denotes the set of partial derivatives $p_\sigma^j= \frac{\partial^{|\sigma|}u^j}{\partial x_1^{i_1} \ldots \partial x_n^{i_n}}$, $\sigma = (i_1,\ldots,i_n)$ ranges over multi-indices with $|\sigma|=i_1+\ldots+i_n$ not exceeding the order of the system $\mathcal E$, while $F_i$ are smooth functions. A symmetry of $\mathcal E$ can be uniquely represented by an evolutionary vector field
\begin{equation}
\label{evf}
\mathbf E_{\Phi} = \sum \limits_{u_\sigma^j \in\mathbb{I}} \bar D_{\sigma}(\phi_j)\frac{\partial}{\partial u_\sigma^j}, \quad \phi_j\in\mathcal{F}(\mathcal{E}),
\end{equation}
where $\Phi=(\phi_1,\phi_2, \ldots, \phi_m)$ is the so-called generating function of the symmetry, ${\mathcal{F}}(\mathcal{E})$ denotes the algebra of smooth differential functions on $\mathcal{E}$ and
\[
\bar D_i=\restr{D_i}{\mathcal{E}},\quad i=1,\dots,n,
\]
are the restrictions of the total derivative operators.

The generating functions $\Phi$ are solutions to the linearized equation
\begin{equation}
\label{lin_sym_cond}
\restr{\ell_F}{\mathcal E} (\Phi) = 0,
\end{equation}
where
$$
\ell_F=\begin{pmatrix}
\sum_\sigma \frac{\partial F_1}{\partial u_\sigma^1}D_\sigma & \sum_\sigma \frac{\partial F_1}{\partial u_\sigma^2}D_\sigma & \ldots \\[3mm]
\sum_\sigma \frac{\partial F_2}{\partial u_\sigma^1}D_\sigma & \sum_\sigma \frac{\partial F_2}{\partial u_\sigma^2}D_\sigma & \ldots\\
\vdots & \vdots & \ddots
\end{pmatrix}
$$
denotes the linearization of $F=(F_1, \ldots, F_k)$. With the commutator as the Lie bracket, the set of all symmetries of $\mathcal E$ forms a Lie algebra $\sym(\mathcal E)$. This Lie algebra is isomorphic to the algebra of all solutions of \eqref{lin_sym_cond}, whose bracket is given by the Jacobi bracket $\jac{\cdot}{\cdot}$ defined by
\begin{equation}
\label{sec1:2}
\jac{\Phi_1}{\Phi_2} = \mathbf E_{\Phi_1}(\Phi_2)-\mathbf E_{\Phi_2}(\Phi_1).
\end{equation}
Accordingly, we will make no distinction between symmetries and their generating functions in what follows.

Applying the recursion operator to a known symmetry amounts to applying it to the corresponding generating function of that symmetry, thereby producing the generating function of a new symmetry. However, one often encounters situations in which the resulting generating function cannot be expressed solely in terms of local variables. For the Lenard operator \eqref{lenard} of the KdV equation, this occurs, for instance, when it is applied to the scaling symmetry; see, e.g., Example 14 in \cite{Kra11}. This naturally leads to extending the original space by introducing additional dependent variables, the so-called nonlocal variables, i.e., passing to a differential covering of the system under consideration.

In this extended setting, symmetries of the original system can no longer be represented in the classical way, i.e., as vector fields on the original equation manifold $\mathcal E$. Within the geometric approach of Krasil'shchik and Vinogradov \cite{Kra89}, they are instead interpreted as so-called shadows of nonlocal symmetries defined on the corresponding nonlocal extension $\tilde{\mathcal E}$. 

Although recursion operators were subsequently interpreted as B\"{a}cklund auto-transformations of the tangent covering of a given system \cite{Pap91, Gut94, Mar95}, this geometric approach does not eliminate their main drawback. Indeed, these operators still act only on shadows of nonlocal symmetries and generate new shadows, rather than full-fledged symmetries at the level of the corresponding covering. Consequently, one must verify in subsequent steps whether the resulting shadows can be lifted to full-fledged symmetries within the given covering.

A qualitatively different approach to recursion operators was introduced in our works \cite{Jah24, Jah25}, where we developed their full-fledged forms for symmetries of the so-called reduced quasi-classical self-dual Yang--Mills equation (rYME)\begin{equation}
\label{yme}
u_{yz} = u_{tx} - u_z u_{xx}+u_xu_{xz}.
\end{equation}
Unlike the conventional approach, the full-fledged recursion operators act directly on full-fledged symmetries in a given covering, rather than merely on their shadows. In \cite{Jah24}, they are constructed as endomorphisms of the $\mathbb{R}$-algebra $\mathcal A^\mathbb{N}$ of all sequences of differential functions in the variables of the corresponding covering. The subset of sequences representing full-fledged symmetries of the rYME \eqref{yme} is invariant under their action. Moreover, these endomorphisms can be represented by matrices of differential functions which, although infinite-dimensional, possess a well-defined and explicitly describable structure. Consequently, applying such recursion operators to full-fledged symmetries reduces to ordinary matrix multiplication, making the procedure computationally straightforward in contrast to the considerably more involved procedures required in the classical framework. In \cite{Jah25}, we further investigated several properties of these full-fledged recursion operators. In particular, we showed how their action on a suitably chosen set of seed full-fledged symmetries can be used to generate the entire $\mathbb{Z}$-graded Lie algebra of full-fledged symmetries of the rYME \eqref{yme}.

One of the central questions in the theory of recursion operators is the verification of their heredity (or hereditary property), a concept first formulated in \cite{Fuch79}. This property plays a crucial role in the construction of commuting hierarchies of symmetries and is one of the characteristic features of the integrability of a given system $\mathcal E$ \cite{Fuc81}.

A fundamental consequence of the hereditary property of recursion operators is the following (see, e.g., \cite{Kra11}): if two symmetries $\Phi_1, \Phi_2$ commute, i.e. $\left\{ \Phi_1, \Phi_2 \right\} = 0$, and the Lie derivatives of the recursion operator $\mathcal{R}$ along these symmetries satisfy
$$\mathcal L_{\Phi_i}(\mathcal R)= \mathbf E_{\Phi_i}(\mathcal R)-[\ell_{\Phi_i}, \mathcal R]=0, \quad i=1,2,$$  
then
$$\left\{ \mathcal R^\alpha (\Phi_1), \mathcal R^\beta (\Phi_2) \right\} = 0$$ 
for all $\alpha,\beta \in \mathbb{N}_0$.

The heredity of a given recursion operator can be characterized by the vanishing of its Nijenhuis torsion $N_{\mathcal R}$, defined by
\begin{equation}
\label{n_tor}
N_{\mathcal R} (\Phi_1,\Phi_2) = \left\{ \mathcal R(\Phi_1), \mathcal R(\Phi_2) \right\} - \mathcal R\left( \left\{ \mathcal R(\Phi_1), \Phi_2 \right\} + \left\{ \Phi_1, \mathcal R(\Phi_2) \right\} - \mathcal R \left\{ \Phi_1, \Phi_2 \right\}  \right).
\end{equation}
However, the explicit computation of the Nijenhuis torsion is, in general, a highly nontrivial task. The Nijenhuis property of recursion operators has been studied in a number of works, in particular within the geometric framework developed by Krasil'shchik and his collaborators; see, e.g., \cite{Ker00, Kra22, Kra23, Mor23, Ver22, Kra17}. These studies provide effective tools for verifying the Nijenhuis property in a number of important classes of integrable systems. Nevertheless, they are formulated at the level of shadows and do not explicitly retain the information encoded in the corresponding full-fledged symmetries within a given differential covering. Consequently, computing the Nijenhuis torsion within the conventional shadow-based representation is not straightforward, since neither the Jacobi bracket nor the action of the recursion operator is intrinsically defined on shadows without reference to the corresponding nonlocal extension.

The main aim of the present paper is to demonstrate that the full-fledged form of a recursion operator, acting directly on full-fledged nonlocal symmetries rather than merely on their shadows, provides a natural and efficient framework for the computation of its Nijenhuis torsion by means of a direct application of the defining formula \eqref{n_tor}. The key observation is that the full-fledged form of such an operator encodes sufficient information about its action on the corresponding covering so that all Jacobi brackets appearing in \eqref{n_tor} can be computed explicitly at the level of generating functions of full-fledged symmetries. The same approach can also be applied to the computation of the Fr\"olicher--Nijenhuis bracket of two recursion operators, providing a direct way to verify their compatibility within the given covering.

As a demonstration of this approach, we consider two four-dimensional linearly degenerate equations, namely, the four-dimensional Mart\'{i}nez Alonso--Shabat equation (4D MAS) \cite{Sha04}
\begin{equation}
\label{4DMAS}
u_{ty} = u_z u_{xy}-u_y u_{xz},
\end{equation}
and the four-dimensional universal hierarchy equation (4D UHE) \cite{Bog17}
\begin{equation}
\label{4DUHE}
u_{yz} = u_{tt}-u_z u_{tx}+u_t u_{xz}.
\end{equation}
Like the rYME \eqref{yme}, both equations arise in a variety of contexts in the literature. For instance, in \cite{Kra23} they are obtained as symmetry reductions of the five-dimensional Mart\'{i}nez Alonso--Shabat equation, while in \cite{Dou19} they arise within the classification of four-dimensional integrable systems associated with sixfolds in the Grassmannian $\textbf{Gr}(4,6)$. The 4D MAS equation \eqref{4DMAS} has attracted considerable attention in the literature. In particular, it was investigated in detail in \cite{Ser14, Zha17, Kra21} and \cite{Mor14, Voj23}, where two distinct recursion operators acting on shadows of nonlocal symmetries were constructed. By contrast, the 4D UHE \eqref{4DUHE} has been studied less extensively. Nevertheless, several aspects of its integrability have been investigated in \cite{Pav17} and \cite{Mor22}, including the construction of a recursion operator for symmetries.

Despite these developments, the available recursion operators for both equations are formulated only at the level of shadows of nonlocal symmetries. To the best of our knowledge, recursion operators acting directly on full-fledged nonlocal symmetries have not yet been described for either equation. The first goal of the present paper is therefore to construct full-fledged recursion operators for both equations. In the case of the 4D UHE \eqref{4DUHE}, we additionally obtain a second recursion operator for shadows of symmetries, complementary to the one previously reported in \cite{Mor22}.

Having constructed these operators, we use their explicit action on full-fledged nonlocal symmetries to compute their Nijenhuis torsions directly from the defining formula \eqref{n_tor}. We then use the same full-fledged forms to compute the Fr\"olicher--Nijenhuis brackets of the recursion operators and thereby verify their compatibility. This provides a clear demonstration of how the full-fledged form of recursion operators can be employed in practice. For this purpose, we work in differential coverings sufficiently rich to accommodate the nonlocal symmetries under consideration, while remaining simple enough to keep the resulting computations explicit and transparent.

\section{Coverings and Full-Fledged Recursion Operators}
In contrast to recursion operators acting merely on symmetry shadows, full-fledged recursion operators are not determined solely by the underlying differential equation. Their explicit form also depends on the choice of differential covering, which specifies the nonlocal variables needed to realize full-fledged nonlocal symmetries. In this section, we construct such coverings together with the corresponding full-fledged recursion operators for the two equations under consideration. These constructions provide the geometric framework required for the subsequent computation of the corresponding Nijenhuis torsions and Fr\"olicher--Nijenhuis brackets.
\subsection{The 4D Mart\'{i}nez Alonso--Shabat equation}
For the 4D MAS equation \eqref{4DMAS}, we build on several results from \cite{Mor14,Voj23}, with some minor changes in notation and conventions. In particular, we will work in the covering
\begin{align} 
\label{mas_cov} 
\tau^q\colon\tilde{\mathcal{E}}^q\to\mathcal{E}&\quad \left|\begin{array}{l} q_0 = u,\\ q_{i,y} = -u_y q_{i-1,x},\\ q_{i,z} = -u_z q_{i-1,x} + q_{i-1,t}, \quad i \geq 1, \end{array}\right. 
\end{align}
cf. the construction and formula (7) in \cite{Voj23}. 
\begin{remark}
The 4D MAS equation \eqref{4DMAS} admits a richer differential covering involving two additional sequences of nonlocal variables; see \cite{Voj23}. Nevertheless, the covering \eqref{mas_cov} is sufficiently rich to lift the recursion operators for shadows constructed in \cite{Mor14,Voj23} to full-fledged forms, while keeping the subsequent computations of the Nijenhuis torsions more manageable and transparent.
\end{remark}

Consider the $\mathbb{R}$-algebra $\mathcal A^{\mathbb N}$ of all sequences of differential functions in the variables of the covering $\tau^q$ (we refer to \cite{Jah24} for a detailed discussion of this formalism). Then each full-fledged nonlocal symmetry of the 4D MAS equation \eqref{4DMAS} in the covering $\tau^q$ is uniquely represented by the vector-valued generating function $P=[p_0, p_1,p_2, \ldots] \in \mathcal A^{\Bbb N}$, whose components are smooth functions on $\tilde{\mathcal{E}}^q$ satisfying the following conditions for all $\alpha\geq1$:
\begin{align}
\label{4DMAS-shad}
\tilde D_{ty}(p_0) + u_y \tilde D_{xz}(p_0)-u_z\tilde D_{xy}(p_0)+u_{xz}\tilde D_y(p_0)-u_{xy}\tilde D_z(p_0)&=0,\\[2mm]
\label{4DMAS-comp1}
\tilde D_y(p_\alpha) + u_y\tilde D_x(p_{\alpha-1})+q_{\alpha-1,x}\tilde D_y(p_0)&=0,\\[2pt]
\label{4DMAS-comp2}
\tilde D_z(p_\alpha) - \tilde D_t(p_{\alpha-1}) + u_z\tilde D_x(p_{\alpha-1})+q_{\alpha-1,x}\tilde D_z(p_0)&=0.
\end{align}
Here $\tilde D_\sigma$ are extensions of the total derivative operators $D_\sigma$ to the covering $\tau^q$ and $p_0$ is the shadow of the full-fledged $\tau^q$-symmetry $P$.

At the level of shadows, two independent recursion operators have been constructed. The first one, due to Morozov \cite{Mor14}, is defined by
\begin{equation}
\label{4DMAS-ro1}
\begin{aligned}
\mathcal R_1^{sh}: \hat p_{0,y} &= u_{xy}p_0-u_y p_{0,x},\\
\hat p_{0,z} &= u_{xz}p_0-u_z p_{0,x}+ p_{0,t}.
\end{aligned}
\end{equation}
In the covering $\tau^q$, these relations can be integrated explicitly using \eqref{4DMAS-comp1}--\eqref{4DMAS-comp2}, yielding
\begin{equation*}
\hat p_0 = u_x p_0+p_1.
\end{equation*}
A direct calculation further shows that the remaining components of the resulting full-fledged symmetry $\hat P=[\hat p_0, \hat p_1, \hat p_2, \ldots ]$ satisfy
\begin{equation*}
\hat p_j = q_{j,x} p_0+p_{j+1}, \quad j\geq1.
\end{equation*}
Since $q_0=u$ by the definition of the covering $\tau^q$, the two preceding relations can be combined into the single formula
\begin{equation}
\label{4DMAS-ro1-ff}
\mathcal R_1:\ \hat p_j=q_{j,x}p_0+p_{j+1},\qquad j\geq0,
\end{equation}
which gives the full-fledged recursion operator $\mathcal R_1$ in the covering $\tau^q$. Furthermore, the action of this operator can be compactly represented in matrix form as
\begin{equation*}
\hat P^T = \mathcal R_1 \cdot P^T,
\end{equation*}
where ${}^T$ denotes transposition, $\cdot$ stands for matrix multiplication, and $\mathcal R_1$ is represented by the following infinite matrix:
\begin{equation*}
{\mathcal R_1}=
\begin{pmatrix}
q_{0,x} & 1 & 0 & 0 & 0 & 0 & 0 & 0 & 0 & 0 & \ldots\\[2mm]
q_{1,x} & 0 & 1 & 0 & 0 & 0 & 0 & 0 & 0 & 0 & \ldots\\[2mm]
q_{2,x} & 0 & 0 & 1 & 0 & 0 & 0 & 0 & 0 & 0 & \ldots\\[2mm]
q_{3,x} & 0 & 0 & 0 & 1 & 0 & 0 & 0 & 0 & 0 & \ldots \\[2mm]
q_{4,x} & 0 & 0 & 0 & 0 & 1 & 0 & 0 & 0 & 0 & \ldots\\[2mm]
q_{5,x} & 0 & 0 & 0 & 0 & 0 & 1 & 0 & 0 & 0 & \ldots\\[2mm]
q_{6,x} & 0 & 0 & 0 & 0 & 0 & 0 & 1 & 0 & 0 & \ldots\\[2mm]
q_{7,x} & 0 & 0 & 0 & 0 & 0 & 0 & 0 & 1 & 0 & \ldots\\[2mm]
\vdots & \vdots & \vdots & \vdots & \vdots & \vdots & \vdots & \vdots & \vdots & \vdots & \ddots
\end{pmatrix}.\\[3mm]
\end{equation*}

\begin{remark}
In what follows, for all full-fledged recursion operators considered below, we denote their action on a generating function P simply by $\mathcal R \cdot P$, suppressing the transpose symbol for notational convenience.
\end{remark}

In \cite{Voj23}, the recursion operators for the 4D MAS equation \eqref{4DMAS} were further studied, and a new independent recursion operator acting on shadows was constructed. In its conventional form, it is given by
\begin{equation}
\label{4DMAS-ro2}
\begin{aligned}
\mathcal R_2^{sh}: \hat p_{0,y} &= t u_{xy} p_0-t u_y p_{0,x} + z p_{0,y},\\
\hat p_{0,z} &= t u_{xz} p_0- t u_z p_{0,x}+ t p_{0,t} + z p_{0,z} + p_0,
\end{aligned}
\end{equation}
which, in view of equations \eqref{4DMAS-comp1}--\eqref{4DMAS-comp2}, is equivalent to
\begin{equation*}
\hat p_0 = tq_{0,x} p_0 + zp_0 + tp_1.
\end{equation*}

As in the case of $\mathcal R_1^{sh}$, this relation can be lifted to the level of full-fledged symmetries. Using the defining equations \eqref{4DMAS-comp1}--\eqref{4DMAS-comp2}, we obtain the following expressions for the remaining components of the resulting full-fledged $\tau^q$-symmetry $\hat P=[\hat p_0, \hat p_1, \hat p_2, \ldots ]$
\begin{equation*}
\hat p_j = tq_{j,x} p_0 + zp_j + tp_{j+1}, \quad j\geq1.
\end{equation*} 
Combining this with the expression for $\hat p_0$, we obtain
\begin{equation}
\label{4DMAS-ro2-ff}
\mathcal R_2:\ \hat p_j=tq_{j,x}p_0+zp_j+tp_{j+1},\qquad j\geq0.
\end{equation}
Thus, the full-fledged recursion operator $\mathcal R_2$ is represented by the infinite matrix
\begin{equation*}
{\mathcal R_2}=
\begin{pmatrix}
tq_{0,x}+z & t & 0 & 0 & 0 & 0 & 0 & 0 & 0 & 0 & \ldots\\[2mm]
tq_{1,x} & z & t & 0 & 0 & 0 & 0 & 0 & 0 & 0 & \ldots\\[2mm]
tq_{2,x} & 0 & z & t & 0 & 0 & 0 & 0 & 0 & 0 & \ldots\\[2mm]
tq_{3,x} & 0 & 0 & z & t & 0 & 0 & 0 & 0 & 0 & \ldots \\[2mm]
tq_{4,x} & 0 & 0 & 0 & z & t & 0 & 0 & 0 & 0 & \ldots\\[2mm]
tq_{5,x} & 0 & 0 & 0 & 0 & z & t & 0 & 0 & 0 & \ldots\\[2mm]
tq_{6,x} & 0 & 0 & 0 & 0 & 0 & z & t & 0 & 0 & \ldots\\[2mm]
tq_{7,x} & 0 & 0 & 0 & 0 & 0 & 0 & z & t & 0 & \ldots\\[2mm]
\vdots & \vdots & \vdots & \vdots & \vdots & \vdots & \vdots & \vdots & \vdots & \vdots & \ddots
\end{pmatrix}.\\[3mm]
\end{equation*}

\begin{remark}
The above matrix representations immediately yield
\begin{equation}
\label{R1R2}
\mathcal R_2=t\mathcal R_1+z\mathcal I,
\end{equation}
where $\mathcal I$ denotes the infinite identity matrix corresponding to the identity map on the $\Bbb R$-algebra $\mathcal A^{\mathbb N}$. Interestingly, the relation \eqref{R1R2} between the recursion operators $\mathcal R_1$ and $\mathcal R_2$ is exactly the same as that between the operators $\mathcal R^q$ and $\mathcal R^m$ for the rYME \eqref{yme}, see Proposition 2 in \cite{Jah24}. This coincidence may be related to the fact that both the 4D MAS equation \eqref{4DMAS} and the rYME \eqref{yme} belong to the same class of linearly degenerate integrable systems in the classification by Doubrov \textit{et al.} \cite{Dou19}, see Table 2 therein.\\
\end{remark}

\subsection{The 4D universal hierarchy equation}
We now turn to the 4D UHE \eqref{4DUHE}. As in the case of the 4D MAS equation \eqref{4DMAS}, we work in a differential covering containing a sufficiently rich set of nonlocal variables to accommodate the full-fledged forms of the recursion operators. We construct the covering by applying the same general procedure as for the covering \eqref{mas_cov}, namely, by deriving it from the Lax pair; see, e.g., \cite{Kra21,Voj23} and references therein for the corresponding construction in the case of the 4D MAS equation \eqref{4DMAS}.

We start from the Lax pair \cite{Pav17} (cf. also \cite{Dou19}) of the 4D UHE \eqref{4DUHE} in the form
\begin{equation*}
\begin{aligned}
w_t &= u_zw_x+\lambda w_z,\\
w_y &= (u_t+\lambda u_z) w_x + \lambda^2w_z,
\end{aligned}
\end{equation*}
where $\lambda$ is a non-removable parameter. We rewrite it in a slightly modified form as
\begin{equation}
\label{4DUHE_Lax2}
\begin{aligned}
w_t &= u_zw_x+\lambda w_z,\\
w_y &= u_t w_x + \lambda w_t,
\end{aligned}
\end{equation}
and then introduce the formal Laurent expansion
$$w=\sum \limits_{i \in \Bbb Z} \lambda^iw_i.$$
Substitution into \eqref{4DUHE_Lax2} yields the infinite-dimensional covering over the 4D UHE \eqref{4DUHE} given by
\begin{equation}
\label{4DUHE_inf_cov}
\begin{aligned}
w_{i,t} &= u_zw_{i,x}+w_{i-1,z},\\
w_{i,y} &= u_tw_{i,x}+w_{i-1,t}, \quad i \in \Bbb Z.
\end{aligned}
\end{equation}
Next, we impose the normalization $w_0=-x$ and set $w_i=0$ for all $i\geq1$. The resulting reduction of the covering \eqref{4DUHE_inf_cov} takes the following form:
\begin{align*}
w_0 &= -x,\quad w_{-1} = u,\\
w_{i,t} &= w_{i+1,y} - u_t w_{i+1,x},\\
w_{i,z} &= w_{i+1,t} - u_z w_{i+1,x}, \quad i\leq -2.
\end{align*}
Finally, setting $w_i=q_{-i-1}$, we arrive at the Abelian covering of the 4D UHE \eqref{4DUHE}
\begin{align} 
\label{4DUHE_cov} 
\tau^q\colon\tilde{\mathcal{E}}^q\to\mathcal{E}&\quad \left|\begin{array}{l} q_{-1} = -x, \quad q_0 = u,\\ 
q_{i,t} = -u_t q_{i-1,x} + q_{i-1,y},\\ 
q_{i,z} = -u_z q_{i-1,x} + q_{i-1,t}, \quad i \geq 1. \end{array}\right. 
\end{align}
Thus, we shall carry out all computations for the 4D UHE \eqref{4DUHE} in the covering \eqref{4DUHE_cov}. As in the case of the 4D MAS equation \eqref{4DMAS}, we have $q_0=u$ by construction.

Each full-fledged $\tau^q$-symmetries of the 4D UHE \eqref{4DUHE} is uniquely represented by a vector-valued generating function $P=[p_0,p_1,p_2,\ldots]$ whose components satisfy the system of equations
\begin{align}
\label{4DUHE-shad}
\tilde D_{yz}(p_0) - \tilde D_t^2(p_0) - u_t\tilde D_{xz}(p_0)
+ u_z\tilde D_{tx}(p_0) + u_{tx}\tilde D_z(p_0)
- u_{xz}\tilde D_t(p_0) &= 0,\\[2mm]
\label{4DUHE-comp1}
\tilde D_t(p_\alpha) + u_t\tilde D_x(p_{\alpha-1})
- \tilde D_y(p_{\alpha-1})
+ q_{\alpha-1,x}\tilde D_t(p_0) &= 0,\\[2pt]
\label{4DUHE-comp2}
\raisetag{18pt}
\tilde D_z(p_\alpha) + u_z\tilde D_x(p_{\alpha-1})
+ u_t\tilde D_x(p_{\alpha-2}) - \tilde D_y(p_{\alpha-2}) \notag\\[-1mm]
{}+q_{\alpha-1,x}\tilde D_z(p_0)
+q_{\alpha-2,x}\tilde D_t(p_0) &= 0,
\qquad \alpha\geq1,
\end{align}
where $p_{-1}=0$.

To find recursion operators for shadows of symmetries of the 4D UHE \eqref{4DUHE}, we seek a new shadow $\phi$ in the form of a linear combination of the components of $P$ and their derivatives up to some finite order, with coefficient functions depending on the internal coordinates of $\tilde{\mathcal E}^q$, which are to be determined. We then impose the condition that $\phi$ satisfy \eqref{4DUHE-shad}. In the present case, we restrict the general \textit{ansatz} to a linear combination involving only the first three components $p_0$, $p_1$, and $p_2$ of $P$. Under this assumption, we obtain the following result:
\begin{proposition}
If the coefficient functions are allowed to depend on the internal coordinates of $\tilde{\mathcal E}^q$ up to first order, the above ansatz yields the following four shadows of symmetries:
\begin{align*}
\phi_0=p_0, \quad \phi_1=u_xp_0&+p_1, \quad \phi_2=(u_x^2+q_{1,x})p_0+u_xp_1+p_2, \\
\intertext{and}
\phi_3&=z\phi_0+t\phi_1+y\phi_2.\\
\end{align*}
\end{proposition}

Let us discuss the particular results. The first shadow $\phi_0$ corresponds to the identity operator and is therefore of no particular interest here. The second shadow $\phi_1$ gives rise to the recursion operator $\hat{\mathcal R}_1^{sh}$ for shadows whose conventional form is
\begin{equation}
\label{4DUHE-ro1}
\begin{aligned}
\hat{\mathcal R}_1^{sh}: \phi_{1,t} &= u_{tx}p_0-u_t p_{0,x} + p_{0,y},\\
\phi_{1,z} &= u_{xz}p_0-u_z p_{0,x} + p_{0,t}.
\end{aligned}
\end{equation}
This operator coincides with the recursion operator previously constructed in \cite{Mor22}.

The third shadow $\phi_2$ corresponds to the square of the recursion operator $\hat{\mathcal R}_1^{sh}$ defined by $\phi_1$ and hence does not provide any essentially new information. We will see this explicitly below, after constructing the full-fledged forms of the corresponding recursion operators and their matrix representations, see Remark \ref{rem_sq}.

In contrast, the fourth shadow $\phi_3$ gives rise to a recursion operator $\hat{\mathcal R}_2^{sh}$ which, to the best of our knowledge, has not been previously identified and is therefore genuinely new. Moreover, unlike the recursion operators considered above, we are not aware of any natural way to express this recursion operator in a conventional form analogous to \eqref{4DUHE-ro1}, or to the forms \eqref{4DMAS-ro1} and \eqref{4DMAS-ro2} obtained for the 4D MAS equation \eqref{4DMAS}. From this perspective, direct computations with the operator $\hat{\mathcal R}_2^{sh}$ on known shadows appear to be rather inconvenient.

The difficulty of working with the conventional representation of $\hat{\mathcal R}_2^{sh}$ suggests that it is preferable to work directly with full-fledged forms. We therefore construct the full-fledged forms of both $\hat{\mathcal R}_1^{sh}$ and $\hat{\mathcal R}_2^{sh}$ in the covering $\tau^q$, providing a uniform representation suitable for the computations below.

We first consider the operator $\hat{\mathcal R}_1^{sh}$. The relation
\begin{equation}
\label{phi1}
\phi_1=q_{0,x}p_0+p_1
\end{equation}
already determines the action of the full-fledged recursion operator $\hat{\mathcal R}_1$ on the shadow component. Using the defining equations \eqref{4DUHE-comp1}--\eqref{4DUHE-comp2}, one can extend this relation to all components of the resulting $\tau^q$-symmetry $\hat P=[\hat p_0,\hat p_1,\hat p_2,\ldots]$. A direct computation yields the following result.
\begin{proposition}
If $P=\left[p_0, p_1, p_2,\ldots  \right]$ is a full-fledged $\tau^q$-symmetry of the 4D UHE \eqref{4DUHE}, then $\hat P=\left[\hat p_0, \hat p_1, \hat p_2,\ldots  \right]$ defined by the components
\begin{align}
\label{4DUHE_ro1_ff}
\hat p_j &= q_{j,x} p_0 + p_{j+1}, \quad j \geq 0,
\end{align}
is also a full-fledged $\tau^q$-symmetry of the 4D UHE \eqref{4DUHE}. Thus, the full-fledged recursion operator $\hat{\mathcal R}_1$ is represented by the infinite matrix
\begin{equation*}
{\hat{\mathcal R}_1}=
\begin{pmatrix}
q_{0,x} & 1 & 0 & 0 & 0 & 0 & 0 & 0 & 0 & 0 & \ldots\\[2mm]
q_{1,x} & 0 & 1 & 0 & 0 & 0 & 0 & 0 & 0 & 0 & \ldots\\[2mm]
q_{2,x} & 0 & 0 & 1 & 0 & 0 & 0 & 0 & 0 & 0 & \ldots\\[2mm]
q_{3,x} & 0 & 0 & 0 & 1 & 0 & 0 & 0 & 0 & 0 & \ldots \\[2mm]
q_{4,x} & 0 & 0 & 0 & 0 & 1 & 0 & 0 & 0 & 0 & \ldots\\[2mm]
q_{5,x} & 0 & 0 & 0 & 0 & 0 & 1 & 0 & 0 & 0 & \ldots\\[2mm]
q_{6,x} & 0 & 0 & 0 & 0 & 0 & 0 & 1 & 0 & 0 & \ldots\\[2mm]
q_{7,x} & 0 & 0 & 0 & 0 & 0 & 0 & 0 & 1 & 0 & \ldots\\[2mm]
\vdots & \vdots & \vdots & \vdots & \vdots & \vdots & \vdots & \vdots & \vdots & \vdots & \ddots
\end{pmatrix}.\\[3mm]
\end{equation*}
\end{proposition}

\begin{remark}
Comparing Eq. \eqref{4DUHE_ro1_ff} with Eq. \eqref{4DMAS-ro1-ff} obtained for the 4D MAS equation \eqref{4DMAS}, we immediately see that the corresponding relations have exactly the same form. Thus, the full-fledged recursion operators $\mathcal R_1$ and $\hat{\mathcal R}_1$ are formally identical, the only difference being the meaning of the nonlocal variables $q_i$ in the respective coverings.
\end{remark}

\begin{remark}
\label{rem_sq}
The matrix representation of $\hat{\mathcal R}_1$ also allows us to compute its square explicitly:
\begin{equation*}
\hat{\mathcal R}_1^2= \hat{\mathcal R}_1 \cdot \hat{\mathcal R}_1=
\begin{pmatrix}
q_{0,x}^2+q_{1,x} & q_{0,x} & 1 & 0 & 0 & 0 & 0 & 0 & 0 & 0 & \ldots\\[2mm]
q_{1,x}u_x+q_{2,x} & q_{1,x} & 0 & 1 & 0 & 0 & 0 & 0 & 0 & 0 & \ldots\\[2mm]
q_{2,x}u_x+q_{3,x} & q_{2,x} & 0 & 0 & 1 & 0 & 0 & 0 & 0 & 0 & \ldots\\[2mm]
q_{3,x}u_x+q_{4,x} & q_{3,x} & 0 & 0 & 0 & 1 & 0 & 0 & 0 & 0 & \ldots \\[2mm]
q_{4,x}u_x+q_{5,x} & q_{4,x} & 0 & 0 & 0 & 0 & 1 & 0 & 0 & 0 & \ldots\\[2mm]
q_{5,x}u_x+q_{6,x} & q_{5,x} & 0 & 0 & 0 & 0 & 0 & 1 & 0 & 0 & \ldots\\[2mm]
q_{6,x}u_x+q_{7,x} & q_{6,x} & 0 & 0 & 0 & 0 & 0 & 0 & 1 & 0 & \ldots\\[2mm]
q_{7,x}u_x+q_{8,x} & q_{7,x} & 0 & 0 & 0 & 0 & 0 & 0 & 0 & 1 & \ldots\\[2mm]
\vdots & \vdots & \vdots & \vdots & \vdots & \vdots & \vdots & \vdots & \vdots & \vdots & \ddots
\end{pmatrix}
\end{equation*}
In particular, the first row of this matrix shows that the shadow component of $\hat{\mathcal R}_1^2 \cdot P$ is
$$
(q_{0,x}^2+q_{1,x})p_0+q_{0,x}p_1+p_2,
$$
which is precisely the shadow $\phi_2$ from Proposition 1. Thus, the shadow $\phi_2$ is precisely the shadow of $\hat{\mathcal R}_1^2\cdot P$, confirming that it is generated by the square $(\hat{\mathcal R}_1^{sh})^2$.
\end{remark}

The form of the shadow $\phi_3$ from Proposition 1 naturally suggests the relation
\begin{equation}
\label{R1R2-UHE}
\hat{\mathcal R}_2=y\hat{\mathcal R}_1^2+t\hat{\mathcal R}_1+z\mathcal I,
\end{equation}
where $\mathcal I$ denotes the identity operator. Combining the matrix representations of $\hat{\mathcal R}_1$ and $\hat{\mathcal R}_1^2$ obtained above then yields the following matrix representation of $\hat{\mathcal R}_2$:
\begin{equation*}
\hat{\mathcal R}_2=
\begin{pmatrix}
y(q_{0,x}^2+q_{1,x})+tq_{0,x}+z & yq_{0,x}+t & y & 0 & 0 & 0 & 0 & 0 & 0 & 0 & \ldots\\[2mm]
y(q_{1,x}q_{0,x}+q_{2,x})+tq_{1,x} & yq_{1,x}+z & t & y & 0 & 0 & 0 & 0 & 0 & 0 & \ldots\\[2mm]
y(q_{2,x}q_{0,x}+q_{3,x})+tq_{2,x} & yq_{2,x} & z & t & y & 0 & 0 & 0 & 0 & 0 & \ldots\\[2mm]
y(q_{3,x}q_{0,x}+q_{4,x})+tq_{3,x} & yq_{3,x} & 0 & z & t & y & 0 & 0 & 0 & 0 & \ldots\\[2mm]
y(q_{4,x}q_{0,x}+q_{5,x})+tq_{4,x} & yq_{4,x} & 0 & 0 & z & t & y & 0 & 0 & 0 & \ldots\\[2mm]
y(q_{5,x}q_{0,x}+q_{6,x})+tq_{5,x} & yq_{5,x} & 0 & 0 & 0 & z & t & y & 0 & 0 & \ldots\\[2mm]
y(q_{6,x}q_{0,x}+q_{7,x})+tq_{6,x} & yq_{6,x} & 0 & 0 & 0 & 0 & z & t & y & 0 & \ldots\\[2mm]
y(q_{7,x}q_{0,x}+q_{8,x})+tq_{7,x} & yq_{7,x} & 0 & 0 & 0 & 0 & 0 & z & t & y & \ldots\\[2mm]
\vdots & \vdots & \vdots & \vdots & \vdots & \vdots & \vdots & \vdots & \vdots & \vdots & \ddots
\end{pmatrix}.
\end{equation*}
The resulting operator $\hat {\mathcal R}_2$ can be verified directly to map full-fledged $\tau^q$-symmetries of the 4D UHE \eqref{4DUHE} to full-fledged $\tau^q$-symmetries by substituting its component formulas
\begin{align}
\label{4DUHE_ro2_ff}
\hat {\mathcal R}_2: \ \hat p_j = (yq_{j,x}q_{0,x}+yq_{j+1,x}+tq_{j,x})p_0+yq_{j,x}p_1+zp_j+tp_{j+1}+yp_{j+2}, \quad j \geq 0,
\end{align}
into the defining equations \eqref{4DUHE-shad}--\eqref{4DUHE-comp2}. This confirms that the operator $\hat{\mathcal R}_2$ defined by \eqref{R1R2-UHE} is a full-fledged recursion operator for the 4D UHE \eqref{4DUHE}.

\begin{remark}
Comparing the relation \eqref{R1R2-UHE} with the corresponding relation \eqref{R1R2} obtained for the 4D MAS equation \eqref{4DMAS} and the rYME \eqref{yme}, we observe that \eqref{R1R2-UHE} differs by the additional term $y\mathcal R_1^2$. This difference may be related to the fact that the 4D UHE \eqref{4DUHE} belongs to a different class in the classification of linearly degenerate integrable systems by Doubrov \textit{et al.} \cite{Dou19} than the 4D MAS equation \eqref{4DMAS} and the rYME \eqref{yme}, see Tables 2 and 3 therein.
\end{remark}

\section{Computation of the Nijenhuis torsion}

The full-fledged form of recursion operators developed above was motivated not only by the need to describe their action on full-fledged symmetries, but also by a more fundamental question: how can one effectively compute the Nijenhuis torsion of such operators, especially in the multidimensional case? Since the vanishing of the Nijenhuis torsion is closely related to the hereditary property of a recursion operator and, consequently, to the generation of hierarchies of commuting symmetries, its explicit computation is of particular importance. For recursion operators of integrable equations with more than two independent variables, however, this question is far from straightforward.

We regard the full-fledged recursion operators as endomorphisms of the $\mathbb R$-algebra $\mathcal A^{\mathbb N}$, which contains the set of generating functions of all full-fledged $\tau^q$-symmetries as an invariant subset. We compute their Nijenhuis torsion on the entire algebra $\mathcal A^{\mathbb N}$. This is sufficient for our purposes: if the Nijenhuis torsion vanishes on all of $\mathcal A^{\mathbb N}$, then its restriction to the invariant subset of generating functions of $\tau^q$-symmetries vanishes as well. We shall apply this approach to the recursion operators obtained for the 4D MAS equation \eqref{4DMAS} and the 4D UHE \eqref{4DUHE}.

\begin{proposition}
\label{R1_heredity}
The full-fledged recursion operators $\mathcal R_1$ for the 4D MAS equation \eqref{4DMAS} and $\hat{\mathcal R}_1$ for the 4D UHE \eqref{4DUHE} are hereditary, i.e., their Nijenhuis torsion \eqref{n_tor} vanishes identically on the $\mathbb R$-algebra $\mathcal A^{\mathbb N}$.
\end{proposition}

\begin{proof}
Since the operators $\mathcal R_1$ and $\hat{\mathcal R}_1$ have the same form and the defining equations for the corresponding nonlocal variables do not enter the computation, it suffices to show that the Nijenhuis torsion \eqref{n_tor} vanishes for one of them.

Let
$$P=[p_0,p_1,p_2,\ldots],\qquad S=[s_0,s_1,s_2,\ldots]$$
be arbitrary elements of $\mathcal A^{\Bbb N}$. Recall that throughout this paper, the components of all elements of $\mathcal A^{\Bbb N}$ are indexed from $0$. We fix $j\geq 0$ and show that the $j$-th component $N_{\mathcal R_1}(P,S)_j$ of the Nijenhuis torsion vanishes identically. 

By \eqref{4DMAS-ro1-ff} we have
$$(\mathcal R_1 \cdot P)_j=q_{j,x}p_0+p_{j+1},$$
where we use the convention $q_0=u$ introduced in \eqref{mas_cov}. Using \eqref{evf} and \eqref{sec1:2}, we obtain
\begin{align*}
\left\{ P,S \right\}_j = \sum \limits_{\sigma, k} \left( \tilde D_\sigma p_k \pd{s_j}{q_{k,\sigma}} - \tilde D_\sigma s_k \pd{p_j}{q_{k,\sigma}} \right),
\end{align*}
where $\tilde D_\sigma$ denotes the corresponding total derivative operator extended to the covering \eqref{mas_cov}. Throughout the computations that follow, the summation indices $\sigma$ and $k$ range over all multiindices and all integers $k\geq 0$, respectively.

Applying \eqref{4DMAS-ro1-ff} to the Jacobi bracket $\left\{P,S\right\}$, we obtain
\begin{equation}
\label{R_PS}
\begin{aligned}
(\mathcal R_1 \cdot \left\{ P,S \right\})_j &= q_{j,x} \sum \limits_{\sigma, k} \left( \tilde D_\sigma p_k \pd{s_0}{q_{k,\sigma}} - \tilde D_\sigma s_k \pd{p_0}{q_{k,\sigma}} \right) + \sum \limits_{\sigma, k} \left( \tilde D_\sigma p_k \pd{s_{j+1}}{q_{k,\sigma}} - \tilde D_\sigma s_k \pd{p_{j+1}}{q_{k,\sigma}} \right)\\[2mm]
&= \sum \limits_{\sigma, k} \tilde D_\sigma p_k \left( q_{j,x} \pd{s_0}{q_{k,\sigma}} + \pd{s_{j+1}}{q_{k,\sigma}} \right) - \tilde D_\sigma s_k \left(q_{j,x} \pd{p_0}{q_{k,\sigma}} + \pd{p_{j+1}}{q_{k,\sigma}} \right)\\[2mm]
&= \sum \limits_{\sigma, k} \tilde D_\sigma p_k \left( \mathcal R_1 \cdot \pd{S}{q_{k,\sigma}} \right)_j - \sum \limits_{\sigma, k} \tilde D_\sigma s_k \left( \mathcal R_1 \cdot \pd{P}{q_{k,\sigma}} \right)_j.
\end{aligned}
\end{equation}

Similarly, using the above formula for $(\mathcal R_1 \cdot P)_j$, we have
\begin{equation}
\label{RP_S}
\begin{aligned}
\left\{\mathcal R_1 \cdot P,S \right\}_j =&  \sum \limits_{\sigma, k} \left( \tilde D_\sigma (\mathcal R_1 \cdot P)_k \pd{s_j}{q_{k,\sigma}} -  \tilde D_\sigma s_k \pd{(\mathcal R_1 \cdot P)_j}{q_{k,\sigma}} \right)\\[2mm]
=&\sum \limits_{\sigma, k} \tilde D_\sigma (\mathcal R_1 \cdot P)_k \pd{s_j}{q_{k,\sigma}} - \sum \limits_{\sigma, k} \tilde D_\sigma s_k \pd{(q_{j,x}p_0+p_{j+1})}{q_{k,\sigma}}\\[2mm]
=& \sum \limits_{\sigma, k} \tilde D_\sigma (\mathcal R_1 \cdot P)_k \pd{s_j}{q_{k,\sigma}} - \sum \limits_{\sigma, k} \tilde D_\sigma s_k \left( q_{j,x}\pd{p_0}{q_{k,\sigma}} + \pd{p_{j+1}}{q_{k,\sigma}}\right) - p_0 \tilde D_x s_j\\[2mm]
=& \sum \limits_{\sigma, k} \tilde D_\sigma (\mathcal R_1 \cdot P)_k \pd{s_j}{q_{k,\sigma}} - \sum \limits_{\sigma, k} \tilde D_\sigma s_k \left( \mathcal R_1 \cdot \pd{P}{q_{k,\sigma}} \right)_j - p_0 \tilde D_x s_j,
\end{aligned}
\end{equation}
resp.
\begin{equation}
\label{P_RS}
\begin{aligned}
\left\{P,\mathcal R_1 \cdot S \right\}_j =&  \sum \limits_{\sigma, k} \left( \tilde D_\sigma p_k \pd{(\mathcal R_1 \cdot S)_j}{q_{k,\sigma}} -  \tilde D_\sigma (\mathcal R_1 \cdot S)_k \pd{p_j}{q_{k,\sigma}} \right)\\[2mm]
=& \sum \limits_{\sigma, k} \tilde D_\sigma p_k \pd{(q_{j,x}s_0+s_{j+1})}{q_{k,\sigma}} -  \sum \limits_{\sigma, k} \tilde D_\sigma (\mathcal R_1 \cdot S)_k \pd{p_j}{q_{k,\sigma}}\\[2mm]
=& \sum \limits_{\sigma, k} \tilde D_\sigma p_k \left( q_{j,x}\pd{s_0}{q_{k,\sigma}} + \pd{s_{j+1}}{q_{k,\sigma}}\right) + s_0 \tilde D_x p_j -  \sum \limits_{\sigma, k} \tilde D_\sigma (\mathcal R_1 \cdot S)_k \pd{p_j}{q_{k,\sigma}}\\[2mm]
=& \sum \limits_{\sigma, k} \tilde D_\sigma p_k \left( \mathcal R_1 \cdot \pd{S}{q_{k,\sigma}} \right)_j - \sum \limits_{\sigma, k} \tilde D_\sigma (\mathcal R_1 \cdot S)_k \pd{p_j}{q_{k,\sigma}} + s_0 \tilde D_x p_j.
\end{aligned}
\end{equation}

Thus, combining \eqref{R_PS}--\eqref{P_RS}, we arrive at
\begin{equation}
\label{RP_S+P_RS-R_PS}
\begin{aligned}
& \left\{\mathcal R_1 \cdot P,S \right\}_j + \left\{P,\mathcal R_1\cdot S \right\}_j - (\mathcal R_1\cdot \left\{ P,S \right\})_j\\[2mm] 
& \hspace{5mm }=  \sum \limits_{\sigma, k} \tilde D_\sigma (\mathcal R_1 \cdot P)_k \pd{s_j}{q_{k,\sigma}}  -
\sum \limits_{\sigma, k}  \tilde D_\sigma (\mathcal R_1 \cdot S)_k \pd{p_j}{q_{k,\sigma}} - p_0 \tilde D_x s_j  + s_0 \tilde D_x p_j,
\end{aligned}
\end{equation}
and applying $\mathcal R_1$ to the expression in \eqref{RP_S+P_RS-R_PS}, we obtain
\begin{equation}
\label{R_RP_S+P_RS-R_PS}
\begin{aligned}
& \left( \mathcal R_1 \cdot \left( \left\{\mathcal R_1\cdot P,S \right\} + \left\{P,\mathcal R_1\cdot S \right\} - \mathcal R_1\cdot \left\{ P,S \right\} \right) \right)_j\\[2mm] 
&\hspace{5mm} = q_{j,x} \left( \sum \limits_{\sigma, k} \tilde D_\sigma (\mathcal R_1\cdot P)_k \pd{s_0}{q_{k,\sigma}}  -
\sum \limits_{\sigma, k}  \tilde D_\sigma (\mathcal R_1 \cdot S)_k \pd{p_0}{q_{k,\sigma}} - p_0 \tilde D_x s_0 + s_0 \tilde D_x p_0 \right)\\[2mm]
& \hspace{10mm} + \sum \limits_{\sigma, k} \tilde D_\sigma (\mathcal R_1 \cdot P)_k \pd{s_{j+1}}{q_{k,\sigma}} -
\sum \limits_{\sigma, k}  \tilde D_\sigma (\mathcal R_1 \cdot S)_k \pd{p_{j+1}}{q_{k,\sigma}}  - p_0 \tilde D_x s_{j+1} + s_0 \tilde D_x p_{j+1}\\[2mm]
&\hspace{5mm} = \sum \limits_{\sigma, k} \tilde D_\sigma (\mathcal R_1  \cdot P)_k \left(q_{j,x} \pd{s_0}{q_{k,\sigma}} +  \pd{s_{j+1}}{q_{k,\sigma}} \right)  -
\sum \limits_{\sigma, k}  \tilde D_\sigma (\mathcal R_1\cdot S)_k \left( q_{j,x} \pd{p_0}{q_{k,\sigma}} +  \pd{p_{j+1}}{q_{k,\sigma}} \right)\\[2mm] 
&\hspace{10mm} - p_0 \left(q_{j,x} \tilde D_x s_0 + \tilde D_x s_{j+1}\right) + s_0 \left(q_{j,x} \tilde D_x p_0 + \tilde D_x p_{j+1}\right)\\[2mm]
&\hspace{5mm} = \sum \limits_{\sigma, k} \tilde D_\sigma (\mathcal R_1  \cdot P)_k \left(\mathcal R_1 \cdot  \pd{S}{q_{k,\sigma}} \right)_j  -
\sum \limits_{\sigma, k}  \tilde D_\sigma (\mathcal R_1\cdot S)_k  \left(\mathcal R_1 \cdot  \pd{P}{q_{k,\sigma}} \right)_j\\[2mm] 
&\hspace{10mm} - p_0 \left( \mathcal R_1 \cdot \tilde D_x S \right)_j + s_0 \left( \mathcal R_1 \cdot \tilde D_x P \right)_j.
\end{aligned}
\end{equation}
Here and below, total derivatives of elements of $\mathcal A^{\mathbb N}$ are understood componentwise.

On the other hand, a direct computation gives
\begin{equation}
\label{RP_RS}
\begin{aligned}
&\left\{ \mathcal R_1\cdot P,\mathcal R_1\cdot S \right\}_j =  \sum \limits_{\sigma, k} \left( \tilde D_\sigma (\mathcal R_1 \cdot P)_k \pd{(\mathcal R_1 \cdot S)_j}{q_{k,\sigma}} - \tilde D_\sigma (\mathcal R_1 \cdot S)_k \pd{(\mathcal R_1 \cdot P)_j}{q_{k,\sigma}} \right)\\[2mm]
&\hspace{5mm} = \sum \limits_{\sigma, k} \tilde D_\sigma (\mathcal R_1 \cdot P)_k \pd{(q_{j,x}s_0+s_{j+1})}{q_{k,\sigma}} - \sum \limits_{\sigma, k} \tilde D_\sigma (\mathcal R_1\cdot S)_k \pd{(q_{j,x}p_0+p_{j+1})}{q_{k,\sigma}}\\[2mm]
&\hspace{5mm} =  \sum \limits_{\sigma, k} \tilde D_\sigma (\mathcal R_1 \cdot P)_k \left(q_{j,x} \pd{s_0}{q_{k,\sigma}} + \pd{s_{j+1}}{q_{k,\sigma}}\right) + s_0\tilde D_x (\mathcal R_1\cdot P)_j\\[2mm]
&\hspace{10mm} - \sum \limits_{\sigma, k} \tilde D_\sigma (\mathcal R_1 \cdot S)_k \left(q_{j,x} \pd{p_0}{q_{k,\sigma}} + \pd{p_{j+1}}{q_{k,\sigma}}\right) - p_0\tilde D_x (\mathcal R_1 \cdot S)_j\\[2mm]
&\hspace{5mm} = \sum \limits_{\sigma, k} \tilde D_\sigma (\mathcal R_1  \cdot P)_k \left(\mathcal R_1 \cdot  \pd{S}{q_{k,\sigma}} \right)_j  - \sum \limits_{\sigma, k}  \tilde D_\sigma (\mathcal R_1\cdot S)_k  \left(\mathcal R_1 \cdot  \pd{P}{q_{k,\sigma}} \right)_j\\[2mm]
&\hspace{10mm}  - p_0\tilde D_x \left(\mathcal R_1 \cdot S\right)_j + s_0\tilde D_x \left(\mathcal R_1\cdot P\right)_j.
\end{aligned}
\end{equation}

The comparison of \eqref{R_RP_S+P_RS-R_PS} and \eqref{RP_RS} reduces the proof to the identity
$$- p_0\tilde D_x \left(\mathcal R_1 \cdot S\right)_j + s_0\tilde D_x \left(\mathcal R_1\cdot P\right)_j = - p_0 \left( \mathcal R_1 \cdot \tilde D_x S \right)_j + s_0 \left( \mathcal R_1 \cdot \tilde D_x P \right)_j$$
which follows directly from the component formula \eqref{4DMAS-ro1-ff} for $\mathcal R_1$. Indeed
\begin{equation}
\begin{aligned}
&- p_0\tilde D_x \left(\mathcal R_1 \cdot S\right)_j + s_0\tilde D_x \left(\mathcal R_1\cdot P\right)_j = -p_0 \tilde D_x \left(q_{j,x} s_0 + s_{j+1}\right) + s_0\tilde D_x \left(q_{j,x} p_0 + p_{j+1}\right) \\[2mm]
& \hspace{5mm} =-q_{j,xx}p_0 s_0 - p_0 (q_{j,x} \tilde D_x s_0 + \tilde D_x s_{j+1}) + q_{j,xx}p_0 s_0 + s_0 (q_{j,x} \tilde D_x p_0 + \tilde D_x p_{j+1})\\
& \hspace{5mm} = - p_0 \left( \mathcal R_1 \cdot \tilde D_x S \right)_j + s_0 \left( \mathcal R_1 \cdot \tilde D_x P \right)_j.
\end{aligned}
\end{equation}

Hence,
$$ N_{\mathcal R_1}(P,S)_j=0.$$
Since $P,S\in\mathcal A^{\Bbb N}$ and $j\geq0$ were arbitrary, the Nijenhuis torsion $N_{\mathcal R_1}$ vanishes identically on the $\mathbb R$-algebra $\mathcal A^{\Bbb N}$.
\end{proof}

We next apply the same method to the second pair of full-fledged recursion operators.

\begin{proposition}
\label{R2_heredity}
The full-fledged recursion operators 
$$\mathcal R_2 = t\mathcal R_1+z\mathcal I$$ 
for the 4D MAS equation \eqref{4DMAS} and 
$$\hat{\mathcal R}_2 = y \hat{\mathcal R}_1^2+t \hat{\mathcal R}_1+z\mathcal I$$
for the 4D UHE \eqref{4DUHE} are hereditary, i.e., their Nijenhuis torsion \eqref{n_tor} vanishes identically on the $\mathbb R$-algebra $\mathcal A^{\mathbb N}$.
\end{proposition}

\begin{proof}
The proof for the operator $\mathcal R_2$ proceeds in the same way as that for $\mathcal R_1$ in Proposition \ref{R1_heredity}. From its component formula \eqref{4DMAS-ro2-ff}, we have
$$(\mathcal R_2 \cdot P)_j=tq_{j,x}p_0+zp_j+tp_{j+1}.$$
Repeating the above calculation, the vanishing of the Nijenhuis torsion reduces to the identity
\begin{equation}
- p_0 \tilde D_x (\mathcal R_2 \cdot S)_j + s_0 \tilde D_x (\mathcal R_2 \cdot P)_j
=
- p_0 (\mathcal R_2 \cdot \tilde D_x S)_j
+ s_0 (\mathcal R_2 \cdot \tilde D_x P)_j,
\end{equation}
which is verified directly from the component formula for $\mathcal R_2$, in complete analogy with the corresponding identity for $\mathcal R_1$.

The same calculation can be carried out for $\hat{\mathcal R}_2$, although its more complicated component formula leads to a less immediate verification. According to \eqref{4DUHE_ro2_ff}, we have
\begin{equation}
\label{RR2_comp}
(\hat{\mathcal R}_2 \cdot P)_j=(yq_{j,x}q_{0,x}+yq_{j+1,x}+tq_{j,x})p_0+yq_{j,x}p_1+zp_j+tp_{j+1}+yp_{j+2}.
\end{equation}
Substituting this expression into the calculation above reduces the verification of the vanishing of the Nijenhuis torsion $N_{\hat{\mathcal R}_2}$ to the more involved identity
\begin{equation}
\label{RR2_identity}
L_1+L_2+L_3 = K_1+K_2+K_3,
\end{equation}
where
\begin{align*}
L_1 &= yq_{j,x}\left( s_0 \tilde D_x (\hat{\mathcal R}_2 \cdot P)_0 - p_0 \tilde D_x (\hat{\mathcal R}_2 \cdot S)_0  \right),\\[2mm]
L_2 &= (yq_{0,x}s_0+ts_0+ys_1) \tilde D_x (\hat{\mathcal R}_2 \cdot P)_j - (yq_{0,x}p_0+tp_0+yp_1) \tilde D_x (\hat{\mathcal R}_2 \cdot S)_j,\\[2mm]
L_3 &= y\left(s_0 \tilde D_x (\hat{\mathcal R}_2 \cdot P)_{j+1} - p_0 \tilde D_x (\hat{\mathcal R}_2 \cdot S)_{j+1}\right),\\[2mm]
K_1 &= yq_{j,x}\left( s_0 (\hat{\mathcal R}_2 \cdot \tilde D_x P)_0 - p_0 (\hat{\mathcal R}_2 \cdot \tilde D_x S)_0  \right),\\[2mm]
K_2 &= (yq_{0,x}s_0+ts_0+ys_1) (\hat{\mathcal R}_2 \cdot \tilde D_x P)_j - (yq_{0,x}p_0+tp_0+yp_1) (\hat{\mathcal R}_2 \cdot \tilde D_x S)_j,\\[2mm]
K_3 &= y\left(s_0 (\hat{\mathcal R}_2 \cdot \tilde D_x P)_{j+1} - p_0 (\hat{\mathcal R}_2 \cdot \tilde D_x S)_{j+1}\right).\\[2mm]
\end{align*}

Expanding $L_1$ according to the component formula \eqref{RR2_comp} and applying the product rule, we find
\begin{equation}
\label{L1_identity}
\begin{aligned}
&L_1 = yq_{j,x} s_0 \tilde D_x \left((yq_{0,x}^2+yq_{1,x}+tq_{0,x})p_0+yq_{0,x}p_1+zp_0+tp_{1}+yp_{2}\right)\\[2mm]
&\hspace{10mm} - yq_{j,x} p_0 \tilde D_x \left((yq_{0,x}^2+yq_{1,x}+tq_{0,x})s_0+yq_{0,x}s_1+zs_0+ts_{1}+ys_{2}\right)\\[2mm]
&\hspace{5mm} = yq_{j,x}\left( s_0 (\hat{\mathcal R}_2 \cdot \tilde D_x P)_0 - p_0 (\hat{\mathcal R}_2 \cdot \tilde D_x S)_0  \right),\\[2mm]
&\hspace{10mm} + yq_{j,x} s_0 \left( (2yq_{0,x}q_{0,xx}+yq_{1,xx}+tq_{0,xx})p_0 + yq_{0,xx}p_1\right)\\[2mm]
&\hspace{10mm} - yq_{j,x} p_0 \left( (2yq_{0,x}q_{0,xx}+yq_{1,xx}+tq_{0,xx})s_0 + yq_{0,xx}s_1\right)\\[2mm]
&\hspace{5mm} = K_1 + y^2 q_{j,x} q_{0,xx}(s_0 p_1 - p_0 s_1).
\end{aligned}
\end{equation}

Similarly, we obtain
\begin{equation}
\label{L2_identity}
L_2 = K_2 + y^2 (q_{j,x} q_{0,xx}+q_{j+1,xx})(p_0 s_1 - s_0 p_1),
\end{equation}
and
\begin{equation}
\label{L3_identity}
L_3 = K_3 + y^2 q_{j+1,xx}(s_0 p_1 - p_0 s_1).
\end{equation}

Summing \eqref{L1_identity}--\eqref{L3_identity}, the remaining terms cancel pairwise, yielding \eqref{RR2_identity}.
\end{proof}

\begin{remark}
The combinations $L_1+L_3$ and $K_1+K_3$ arising in the above computation admit more compact representations involving the recursion operator $\hat{\mathcal R}_1$. Namely,
\begin{align*}
L_1+L_3
&=ys_0\left(\hat{\mathcal R}_1\cdot\tilde D_x(\hat{\mathcal R}_2\cdot P)\right)_j
-yp_0\left(\hat{\mathcal R}_1\cdot\tilde D_x(\hat{\mathcal R}_2\cdot S)\right)_j,\\[2mm]
K_1+K_3
&=ys_0\left(\hat{\mathcal R}_1\cdot(\hat{\mathcal R}_2\cdot\tilde D_xP)\right)_j
-yp_0\left(\hat{\mathcal R}_1\cdot(\hat{\mathcal R}_2\cdot\tilde D_xS)\right)_j.
\end{align*}
Nevertheless, these representations are not used in the proof, where the Nijenhuis torsion of $\hat{\mathcal R}_2$ is computed directly from its component formula \eqref{4DUHE_ro2_ff}.
\end{remark}

\section{The Fr\"olicher--Nijenhuis bracket}
The heredity of the full-fledged recursion operators established in the previous section concerns each operator separately. We now turn to the compatibility of the recursion operators associated with the 4D MAS equation \eqref{4DMAS} and the 4D UHE \eqref{4DUHE}. Compatibility is a natural stronger property for a pair of hereditary recursion operators: it guarantees that their linear combinations are again hereditary. This property is characterized by the vanishing of the Fr\"olicher--Nijenhuis bracket, which extends the Nijenhuis torsion to pairs of operators.

Consider a linear combination
$$\mathcal R_1+\lambda\mathcal R_2$$
of two hereditary recursion operators and substitute it into the defining formula for the Nijenhuis torsion \eqref{n_tor}. We obtain
\begin{align*}
N_{\mathcal R_1+\lambda\mathcal R_2}(\Phi_1,\Phi_2)
&=
\{(\mathcal R_1+\lambda\mathcal R_2)\cdot\Phi_1,
(\mathcal R_1+\lambda\mathcal R_2)\cdot\Phi_2\}\\
&\hspace{-15mm}-(\mathcal R_1+\lambda\mathcal R_2)\cdot
\Big(
\{(\mathcal R_1+\lambda\mathcal R_2)\cdot\Phi_1,\Phi_2\}
+\{\Phi_1,(\mathcal R_1+\lambda\mathcal R_2)\cdot\Phi_2\}
-(\mathcal R_1+\lambda\mathcal R_2)\cdot\{\Phi_1,\Phi_2\}
\Big).
\end{align*}
Expanding this expression with respect to $\lambda$, we obtain
\begin{align*}
N_{\mathcal R_1+\lambda\mathcal R_2}(\Phi_1,\Phi_2)=N_{\mathcal R_1}(\Phi_1,\Phi_2)
+\lambda[\mathcal R_1,\mathcal R_2]_{\mathrm{FN}}(\Phi_1,\Phi_2) +\lambda^2N_{\mathcal R_2}(\Phi_1,\Phi_2),
\end{align*}
where the coefficient of $\lambda$ is
\begin{equation}
\label{FN_bracket}
\begin{aligned}
[\mathcal R_1,\mathcal R_2]_{\mathrm{FN}}(\Phi_1,\Phi_2)
&=
\{\mathcal R_1\cdot\Phi_1,\mathcal R_2\cdot\Phi_2\}
+\{\mathcal R_2\cdot\Phi_1,\mathcal R_1\cdot\Phi_2\}\\[2mm]
&\hspace{5mm}-\mathcal R_1\cdot
\left(
\{\mathcal R_2\cdot\Phi_1,\Phi_2\}
+\{\Phi_1,\mathcal R_2\cdot\Phi_2\}
-\mathcal R_2\cdot\{\Phi_1,\Phi_2\}
\right)\\[2mm]
&\hspace{5mm}-\mathcal R_2\cdot
\left(
\{\mathcal R_1\cdot\Phi_1,\Phi_2\}
+\{\Phi_1,\mathcal R_1\cdot\Phi_2\}
-\mathcal R_1\cdot\{\Phi_1,\Phi_2\}
\right),
\end{aligned}
\end{equation}
and we call it the Fr\"olicher--Nijenhuis bracket.

Since both $\mathcal R_1$ and $\mathcal R_2$ are hereditary, their linear combination $\mathcal R_1+\lambda\mathcal R_2$ is hereditary for every $\lambda$ if and only if
\[
[\mathcal R_1,\mathcal R_2]_{\mathrm{FN}}=0.
\]
This motivates the direct computation of the Fr\"olicher--Nijenhuis bracket for the pairs of full-fledged recursion operators studied below.
\begin{proposition}
The recursion operators $\mathcal R_1$ and $\mathcal R_2$ for the 4D MAS equation \eqref{4DMAS} are compatible, i.e. their Fr\"olicher--Nijenhuis bracket $[\mathcal R_1,\mathcal R_2]_{\mathrm{FN}}$ vanishes identically on the $\Bbb R$-algebra $\mathcal A^\Bbb N$.
\end{proposition}
\begin{proof}
As in the proofs of Propositions \ref{R1_heredity} and \ref{R2_heredity}, we use the component formulas \eqref{4DMAS-ro1-ff} for $\mathcal R_1$ and \eqref{4DMAS-ro2-ff} for $\mathcal R_2$ to compute the components of their Fr\"olicher--Nijenhuis bracket directly from the defining formula \eqref{FN_bracket}. For each $j\geq0$, the four terms entering \eqref{FN_bracket} are evaluated as follows:
\begin{align*}
&\left( \mathcal R_1 \cdot \left( \left\{\mathcal R_2\cdot P,S \right\} + \left\{P,\mathcal R_2 \cdot S \right\} - \mathcal R_2\cdot \left\{ P,S \right\} \right) \right)_j \\[2mm]
&\hspace{5mm} =  \sum \limits_{\sigma, k} \tilde D_\sigma (\mathcal R_2 \cdot P)_k \left( \mathcal R_1 \cdot \pd{S}{q_{k,\sigma}} \right)_j - \sum \limits_{\sigma, k} \tilde D_\sigma (\mathcal R_2 \cdot S)_k \left( \mathcal R_1 \cdot \pd{P}{q_{k,\sigma}} \right)_j\\
&\hspace{10mm} - tp_0 (\mathcal R_1 \cdot \tilde D_x S)_j + ts_0 (\mathcal R_1 \cdot \tilde D_x P)_j,
\end{align*}

\begin{align*}
& \left( \mathcal R_2 \cdot \left( \left\{\mathcal R_1\cdot P,S \right\} + \left\{P,\mathcal R_1\cdot S \right\} - \mathcal R_1\cdot \left\{ P,S \right\} \right) \right)_j\\[2mm] 
&\hspace{5mm} = \sum \limits_{\sigma, k} \tilde D_\sigma (\mathcal R_1  \cdot P)_k \left(\mathcal R_2 \cdot  \pd{S}{q_{k,\sigma}} \right)_j  -
\sum \limits_{\sigma, k}  \tilde D_\sigma (\mathcal R_1\cdot S)_k  \left(\mathcal R_2 \cdot  \pd{P}{q_{k,\sigma}} \right)_j\\[2mm] 
&\hspace{10mm} - p_0 \left( \mathcal R_2 \cdot \tilde D_x S \right)_j + s_0 \left( \mathcal R_2 \cdot \tilde D_x P \right)_j,
\end{align*}

\begin{align*}
\left\{\mathcal R_1 \cdot P, \mathcal R_2 \cdot S \right\}_j &=  \sum \limits_{\sigma, k} \tilde D_\sigma (\mathcal R_1 \cdot P)_k \left( \mathcal R_2 \cdot \pd{S}{q_{k,\sigma}} \right)_j + ts_0 \tilde D_x (\mathcal R_1 \cdot P)_j\\
&\hspace{10mm} - \sum \limits_{\sigma, k} \tilde D_\sigma (\mathcal R_2 \cdot S)_k \left( \mathcal R_1 \cdot \pd{P}{q_{k,\sigma}} \right)_j - p_0 \tilde D_x (\mathcal R_2 \cdot S)_j,
\end{align*}

\begin{align*}
\left\{\mathcal R_2 \cdot P, \mathcal R_1 \cdot S \right\}_j &=  \sum \limits_{\sigma, k} \tilde D_\sigma (\mathcal R_2 \cdot P)_k \left( \mathcal R_1 \cdot \pd{S}{q_{k,\sigma}} \right)_j + s_0 \tilde D_x (\mathcal R_2 \cdot P)_j\\
&\hspace{10mm} - \sum \limits_{\sigma, k} \tilde D_\sigma (\mathcal R_1 \cdot S)_k \left( \mathcal R_2 \cdot \pd{P}{q_{k,\sigma}} \right)_j - tp_0 \tilde D_x (\mathcal R_1 \cdot S)_j.
\end{align*}
Thus, the proof reduce to the identity
\begin{align*} 
&-p_0\tilde D_x \left(\mathcal R_2 \cdot S\right)_j + ts_0\tilde D_x \left(\mathcal R_1\cdot P\right)_j -tp_0\tilde D_x \left(\mathcal R_1 \cdot S\right)_j + s_0\tilde D_x \left(\mathcal R_2\cdot P\right)_j\\[2mm] 
&\hspace{5mm}= - p_0 \left( \mathcal R_2 \cdot \tilde D_x S \right)_j + ts_0 \left( \mathcal R_1 \cdot \tilde D_x P \right)_j - tp_0 \left( \mathcal R_1 \cdot \tilde D_x S \right)_j + s_0 \left( \mathcal R_2 \cdot \tilde D_x P \right)_j
\end{align*}
which follows directly from the component formulas \eqref{4DMAS-ro1-ff} and \eqref{4DMAS-ro2-ff}.
\end{proof}

An analogous statement holds for the recursion operators $\hat{\mathcal R}_1$ and $\hat{\mathcal R}_2$ for the 4D UHE \eqref{4DUHE}.
\begin{proposition}
The recursion operators $\hat{\mathcal R}_1$ and $\hat{\mathcal R}_2$ for the 4D UHE \eqref{4DUHE} are compatible, i.e. their Fr\"olicher--Nijenhuis bracket $[\hat{\mathcal R}_1,\hat{\mathcal R}_2]_{\mathrm{FN}}$ vanishes identically on the $\Bbb R$-algebra $\mathcal A^\Bbb N$.
\end{proposition}
\begin{proof}
Using the component formulas \eqref{4DUHE_ro1_ff} and \eqref{4DUHE_ro2_ff} for $\hat{\mathcal R}_1$ and $\hat{\mathcal R}_2$, respectively, the vanishing of their Fr\"olicher--Nijenhuis bracket reduces to the identity
\begin{align*}
&ys_0 (\hat{\mathcal R}_1 \cdot \tilde D_x (\hat{\mathcal R}_1 \cdot P))_j + (yq_{0,x}s_0+ts_0+ys_1) \tilde D_x (\hat{\mathcal R}_1 \cdot P)_j - p_0 \tilde D_x (\hat{\mathcal R}_2 \cdot S)_j\\[2mm]
&\hspace{5mm} - yp_0 (\hat{\mathcal R}_1 \cdot \tilde D_x (\hat{\mathcal R}_1 \cdot S))_j - (yq_{0,x}p_0+tp_0+yp_1) \tilde D_x (\hat{\mathcal R}_1 \cdot S)_j + s_0 \tilde D_x (\hat{\mathcal R}_2 \cdot P)_j \\[2mm]
&= ys_0 (\hat{\mathcal R}_1^2 \cdot \tilde D_x P)_j + (yq_{0,x}s_0+ts_0+ys_1) (\hat{\mathcal R}_1 \cdot \tilde D_x P)_j - p_0 (\hat{\mathcal R}_2 \cdot \tilde D_x S)_j\\[2mm]
&\hspace{5mm} - yp_0 (\hat{\mathcal R}_1^2 \cdot \tilde D_x S )_j  - (yq_{0,x}p_0+tp_0+yp_1) (\hat{\mathcal R}_1 \cdot \tilde D_x S)_j + s_0 (\hat{\mathcal R}_2 \cdot \tilde D_x P)_j,
\end{align*}
which follows directly from the component formulas \eqref{4DUHE_ro1_ff} and \eqref{4DUHE_ro2_ff}.
\end{proof}

\section{Conclusions}
In this paper, we have developed and applied a full-fledged approach to recursion operators for multidimensional linearly degenerate integrable equations. Starting from their action on full-fledged symmetries, we represented the recursion operators as endomorphisms of the $\mathbb R$-algebra $\mathcal A^{\mathbb N}$ and used these representations to investigate their algebraic properties directly. For the 4D MAS equation \eqref{4DMAS} and the 4D UHE \eqref{4DUHE}, we have shown that the corresponding full-fledged recursion operators are hereditary and that the two operators associated with each equation are compatible. In particular, the direct computation of the Nijenhuis torsion and the Fr\"olicher--Nijenhuis bracket demonstrates that the full-fledged representation provides a practical framework for studying these properties in the multidimensional setting.

The approach developed in this paper is not restricted to the particular coverings considered here and can be applied to full-fledged recursion operators represented in more general coverings. As an example, one may consider the richer covering for the 4D MAS equation constructed in \cite{Voj23} in connection with the inversion of the recursion operator $\mathcal R_2^{sh}$ acting on shadows. At the same time, constructing a covering of the 4D UHE equation suitable for describing the inverse of $\hat{\mathcal R}_2$ remains an open problem. The approach is likewise not restricted to the two equations considered here and can in principle be applied to other multidimensional integrable equations as soon as full-fledged forms of their recursion operators are available.

Remarkably, the infinite-dimensional matrix representations of the recursion operators $\mathcal R_1$ and $\hat{\mathcal R}_1$ obtained in this paper are precisely of the first companion form appearing in the theory of $\mathfrak{gl}$-regular Nijenhuis operators developed by Bolsinov, Konyaev, and Matveev \cite{Bol24,Kon24,Mat24}. While their theory has been developed primarily in the finite-dimensional setting, the appearance of this companion structure in the recursion operators of four-dimensional integrable PDEs raises the natural question of whether their approach admits a suitable extension to the present infinite-dimensional setting. Exploring this connection may provide new structural information about full-fledged recursion operators and their associated symmetry and conservation-law hierarchies, and could facilitate the construction of explicit exact solutions through finite-dimensional reductions and integration in quadratures. We intend to pursue these questions in future work.

\section*{Acknowledgments}
\label{sec:acknowledgments}

The symbolic computations were performed using the software \textsc{Jets} \cite{Jets}. The research was supported by the Ministry of Education, Youth and Sports of the Czech Republic (MSMT CR) under RVO funding for IC47813059.

\end{document}